\documentclass[10pt,a4paper]{article}
\usepackage[margin=2.875cm]{geometry}

\usepackage[T1]{fontenc}
\usepackage{microtype}

\usepackage[shortlabels]{enumitem}
\usepackage{booktabs}
\usepackage[font=small,margin=0.1\linewidth]{caption}

\usepackage{graphicx,xcolor}
\usepackage[colorlinks=true,hyperindex=true]{hyperref}
\usepackage{nameref}

\definecolor{newdarkblue}{RGB}{24,64,184}
\definecolor{newgreen}{RGB}{80,170,55}
\hypersetup{
    colorlinks,%
    citecolor=newdarkblue,%
    filecolor=red,%
    linkcolor=newdarkblue,%
    urlcolor=newgreen,%
    pdfnewwindow=true,%
    pdfstartview={FitH}
}

\usepackage{amsmath,amsfonts,amssymb,amsthm,mathtools,bbm}
\usepackage{esint}
\usepackage{braket}

\numberwithin{equation}{section}

\newtheorem{theorem}{Theorem}[section]
    \newtheorem{proposition}[theorem]{Proposition}
    \newtheorem{lemma}[theorem]{Lemma}

\theoremstyle{definition}
    \newtheorem{definition}[theorem]{Definition}
    \newtheorem{remark}[theorem]{Remark}
    \newtheorem{example}[theorem]{Example}

\def\zz{{\mathbb Z}}
\def\nn{{\mathbb N}}

\def\cA{{\cal A}}

\def\P{{\mathbb P}}
\def\Q{{\mathbb Q}}
\newcommand{\Sc}{h_{\textnormal{c}}}
\newcommand{\Sr}{h_{\textnormal{r}}}

\newcommand{\commentOutCode}[1]{}

	\makeatletter
	\providecommand*{\diff}%
	{\@ifnextchar^{\DIfF}{\DIfF^{}}}
	\def\DIfF^#1{%
	\mathop{\mathrm{\mathstrut d}}%
	\nolimits^{#1}\gobblespace}
	\def\gobblespace{%
	\futurelet\diffarg\opspace}
	\def\opspace{%
	\let\DiffSpace\!%
	\ifx\diffarg(%
	\let\DiffSpace\relax
	\else
	\ifx\diffarg%
	\let\DiffSpace\relax
	\else
	\ifx\diffarg\{%
	\let\DiffSpace\relax
	\fi\fi\fi\DiffSpace}

\newcommand{\Exp}[1]{\mathrm{e}^{#1}}

\DeclareMathOperator{\supp}{supp}

\newcommand{\shift}{T}

\makeatletter
\def\namedlabel#1#2{\begingroup
   \def\@currentlabel{#2}%
   \label{#1}\endgroup
}
\makeatother

\begin{document}

\title{On the Ziv--Merhav theorem beyond Markovianity}
\author{N.\ Barnfield\textsuperscript{a}, R.\ Grondin\textsuperscript{a}, G.\ Pozzoli\textsuperscript{b} and R.\ Raqu\'epas\textsuperscript{c}}
\date{}

\maketitle

\begin{center}
  \begin{minipage}[b]{0.5\textwidth}
    \small
    \centering
    a. McGill University \\
    Department of Mathematics and Statistics \\
    Montr\'{e}al QC, Canada
  \end{minipage}%
  \begin{minipage}[b]{0.5\textwidth}
    \small
    \centering
    b. Universit\`{a} degli Studi di Milano-Bicocca  \\
    Dipartimento di Matematica e Applicazioni \\
    Milan, Italy
  \end{minipage} 
  
  \vspace{1em}
  
  \begin{minipage}[b]{0.45\textwidth}
    \small
    \centering
    c. New York University \\
    Courant Institute of Mathematical Sciences \\
    New York NY, United States\\
  \end{minipage}
\end{center}

\begin{abstract}
    We generalize to a broader class of decoupled measures a result of Ziv and Merhav on universal estimation of the specific cross (or relative) entropy for a pair of multi-level Markov measures. The result covers pairs of suitably regular g-measures and pairs of equilibrium measures arising from the ``small space of interactions'' in mathematical statistical mechanics.
\end{abstract}

\begin{center}
    \emph{Dedicated to the memory of Jacob Ziv} (\emph{1931--2023})
\end{center}

\section{Introduction}

In 1993, Ziv and Merhav proposed a ``new notion of empirical informational divergence'', or relative-entropy estimator, based on the celebrated Lempel--Ziv compression algorithm~\cite{MZ93}. While this estimator received\,---\,to our knowledge\,---\,little attention in the mathematical literature, it (and its variants) has met with success in many practical applications across fields such as linguistics, medicine, and physics; see e.g.~\cite{BCL02,CF05,B+08,CFF10,RP12,LMDEC19,R+22}, to only cite a few. In fact, our main motivation for a more extensive rigorous treatment of the convergence of this estimator is that the very limited Markovian class of sources covered by the original result of Ziv and Merhav pales in comparison with the breadth of apparent applicability. 

Ziv and Merhav's sequence of estimators is defined as follows. Given two strings~$x_1^N$ and~$y_1^N$, let $c_N(y|x)$ be the number of words in a sequential parsing of~$y_1^N$ using the longest possible substrings of~$x_1^N$; if there is no such substring of~$x_1^N$, the parsed word is set to be one letter long. For example, if
\begin{align*}
    x &= 01000101110100111001000122021\dots, \\
    y &= 01100101000102011101001000210\dots,
\end{align*}
and $N=24$, then the Ziv--Merhav parsing of 
$
    y_1^{24}=011001010001020111010010
$
with respect to
$
    x_1^{24}=010001011101001110010001
$
is
\[
    y_1^{24} = 011|00101|00010|2|011101001|0
\]
and $c_{24}(y|x)=6$.\footnote{Throughout this paper, we will refer to the partitioning symbol ``$|$'' as a separator, and we will say that a separator \emph{falls within} a given string if the separator lies after one of the letters that make up the string.} 
Ziv and Merhav show that the estimator
\[
    \widehat{Q}_N(y,x) := \frac{c_N(y|x) \ln N}{N}
\]
converges to the specific cross entropy~$\Sc(\Q|\P)$ between the sources~$\P$ and~$\Q$ that have produced~$x$ and~$y$ respectively,
under the assumption that those measures come from irreducible multi-level Markov chains. We will refer to~$(\widehat{Q}_N)_{N=1}^\infty$ as the ZM estimator. The relative entropy~$\Sr(\Q|\P)$ can then be estimated by combining the above with an estimation of the specific entropy~$h(\Q)$, say \emph{à la} Lempel--Ziv~\cite{LZ78}. Both quantities are defined in Section~\ref{sec:setting} for the reader's convenience. Our goal is to generalize this result beyond Markovianity, namely under conditions \ref{it:ID}, \ref{it:KB} and \ref{it:FE} below.

One may note that the behaviour of $c_N$ is intimately related to the so-called Wyner--Ziv problem on \emph{waiting times}. With 
\[ 
    W_\ell(y,x) := \inf\{r \in \nn : x_{r}^{r+\ell} = y_1^\ell\},
\]
the Wyner--Ziv problem concerns the convergence
\begin{equation}
\label{eq:WZ-convergence}
    \frac{\ln W_\ell}{\ell} \to \Sc(\Q|\P)
\end{equation}
as $\ell\to\infty$ within sufficiently nice classes of measures~\cite{WZ89,Sh93,Ko98}. 
To see the relation, note that the length of the first word in the ZM parsing of $y_1^N$ with respect to~$x_1^N$ is\,---\,save some edge cases\,---\,the largest possible $\ell$ such that $W_{\ell}(y,x) \leq N - \ell+1$. This dual quantity is known as the \emph{longest match length}
\[ 
    \Lambda_N(y,x) := \max\left\{ 1, \sup\{\ell\in\nn : W_\ell(y,x) \leq N-\ell+1 \}\right\}.
\]
The length of the second word in this parsing is then\,---\,again save some edge cases handled in Section~\ref{ssec:comments}\,---\,the longest-match length $\Lambda_{N}(T^{\Lambda_N(y,x)}y,x)$, and so on. 
Any attempt at a theory of the asymptotic behaviour of waiting times and its derived quantities beyond Markovianity must take two important caveats into account. First, it is known that the specific cross entropy between two ergodic sources does not always exist; see e.g.~\cite[\S{A.5.2}]{vEFS93}.
Second, it is known that there exists a mixing measure~$\P$ such that~\eqref{eq:WZ-convergence} fails with $\Q=\P$; see~\cite[\S{4}]{Sh93}. While the precise breadth of the validity of~\eqref{eq:WZ-convergence} and its different refinements remains unknown, a focus on decoupling conditions in the spirit of~\cite{Pf02} has recently proved effective for making significant progress~\cite{CDEJR23w,CR23}; the present contribution follows along those lines. 

More generally, the present work is part of a broader research program~\cite{BJPP18,BCJPPExamples,CDEJR23,CDEJR23w,CR23} whose goals include promoting the efficiency of this ``decoupling perspective'' originating in statistical mechanics in revisiting long-standing problems in dynamical systems and information theory. This efficiency concerns both the reformulation of different existing proof strategies in a common language and the generation of nontrivial extensions.

\paragraph*{Organization of the paper.} The rest of the paper is organized as follows. In Section~\ref{sec:setting}, we set the stage by properly introducing our notation, objects of interest, and assumptions. In Section~\ref{sec:simplified}, we state our main result, provide its proof, and make several comments. In Section~\ref{sec:examples}, we discuss examples to which this result applies beyond Markovianity.

\section{Setting}
\label{sec:setting}

Let $\Omega\coloneqq  \{(x_k)_{k\in\nn} : x_k\in\cA\text{ for all }  k\in\nn\}$ be equipped with the $\sigma$-algebra generated by cylinders of the form~$[a] := \{x \in \Omega : x_1^n = a\}$. The \emph{shift} map $T: \Omega\to \Omega$ defined by $(Tx)_k \coloneqq x_{k+1}$ is then a measurable surjection. Let $\P$ and~$\Q$ be stationary (i.e.~$T$-invariant) probability measures  on~$\Omega$.
We set
\[ 
    \supp \P_n := \{a \in \cA^n : \P[a] > 0 \}
\]
and 
\[ 
    \supp \P := \{x \in \Omega : x_1^n \in \supp \P_n \text{ for all } n\in\nn\},
\]
and similarly for~$\Q$. The (specific) \emph{entropy} $h(\P)$ of a measure $\P$ is
\[
    h(\P):=\lim_{n\to\infty}-\frac{1}{n}\sum_{a\in \mathcal{A}^n}\P[a]\ln \P[a].
\]
Fekete's lemma ensures that this limit always exists and lies in~$[0, \ln (\#\cA)]$.
The (specific) \emph{cross entropy} of~$\Q$ with respect to~$\P$ is
\begin{equation}
\label{eq:def-cross}
    \Sc(\Q|\P):=\lim_{n\to\infty}-\frac{1}{n}\sum_{a\in\mathcal{A}^n}\Q[a]\ln \P[a],
\end{equation}
\emph{when the limit exists} in~$[0,\infty]$.
In this case, the (specific) \emph{relative entropy} $\Sr(\Q|\P)$ of $\Q$ with respect to~$\P$ is then defined as
\[
    \Sr(\Q|\P):= \Sc(\Q|\P) - h(\Q).
\]
The abstract properties of stationary measures that we will work with are the following:
\begin{description}
    \item[{ID}\namedlabel{it:ID}{ID}] A measure $\P$ is said to be \emph{immediately decoupled on its support} if there exists a nondecreasing, $o(n)$-sequence $(k_n)_{n=1}^\infty$ such that, for every~$n\in\nn$, both
    \begin{equation}
    \label{eq:ID-upper}
        \sup_{m\in\nn}\max\left\{\frac{\P[ab]}{\P[a]\P[b]} : a\in\supp \P_n, b \in \supp \P_m \right\} \leq \Exp{k_n}
    \end{equation}
    and
    \begin{equation}
    \label{eq:ID-lower}
        \inf_{m\in\nn}\min\left\{\frac{\P[ab]}{\P[a]\P[b]} : a\in\supp \P_n, b \in \supp \P_m,  ab\in\supp \P_{n+m} \right\} \geq \Exp{-k_n}.
    \end{equation}

    \item[{FE}\namedlabel{it:FE}{FE}] The $\P$-measure of cylinders is said to decay \emph{fast enough} if there exists $\gamma_+ < 0$ such that 
    \begin{equation}
        \sup_{a \in \supp \P_n} \P[a] \leq \Exp{\gamma_+ n}
    \end{equation}
    for all $n$ large enough.

    \item[{KB}\namedlabel{it:KB}{KB}] The measure $\P$ is said to satisfy \emph{Kontoyiannis' bound on waiting times} if there exist $o(n)$-sequences $(k_n)_{n=1}^\infty$ and $(\tau_n)_{n=1}^\infty$ such that
    \[
        \P\{x : W_\ell(a,x) \geq r\} 
        \leq 
        \exp \left(-\Exp{-k_\ell}\P[a]\left\lfloor \frac{r-1}{\ell+\tau_\ell}\right\rfloor\right)
    \]
    for every $\ell \in \nn$, $a \in \cA^\ell$ and $r\in\nn$.
\end{description}

Let us briefly discuss these abstract assumptions. 
First, it is straightforward to show that if $\P$ is the stationary measure for an irreducible multi-level Markov chain with positive entropy, then $\P$ satisfies~\ref{it:ID}, \ref{it:FE}, and \ref{it:KB}. Already for Markov chains, we see that only requiring the lower bound~\eqref{eq:ID-lower} when $ab$ is in the support is significant: requiring the lower bound whenever~$a$ and~$b$ are in the support (separately) would be considerably more restrictive, as this would exclude all Markov measures for which some transition probability is null.
Second, the bound \ref{it:KB} was derived in~\cite{Ko98} under a mixing assumption, but the following implication seems more natural for the classes of examples we have in mind: \ref{it:KB} will follow from \ref{it:ID} if one is willing to assume that the support of~$\P$ satisfies\,---\,as a subshift of~$\Omega$\,---\,a suitable notion of specification; see~\cite[\S\S{3.1,\,A.1,\,B.2}]{CR23} and \cite[\S{3}]{CDEJR23w}.
Third, repeated uses of~\eqref{eq:ID-lower} in \ref{it:ID} implies the following property, which naturally complements \ref{it:FE}: 
\begin{description}
    \item[SE\namedlabel{it:SE}{SE}] The $\P$-measure of cylinders is said to decay \emph{slow enough} if there exists $\gamma_- < 0$ such that 
    \[  
        \inf_{b \in \supp \P_n } \P[b] \geq \Exp{\gamma_- n}
    \]
    for all $n$ large enough.
\end{description}
These assumptions are established and discussed in the context of important classes of examples in Section~\ref{sec:examples}.

\section{The main result}
\label{sec:simplified}

\subsection{The statement and structure of the proof}
\label{ssec:proof-outline}

\begin{theorem}
\label{thm:simplified}
    Suppose that the stationary measure~$\P$ satisfies \ref{it:ID}, \ref{it:FE}, and \ref{it:KB} and that the ergodic measure~$\Q$ satisfies \ref{it:ID} and \ref{it:FE}.\footnote{In fact, as far as the decoupling of~$\Q$ is concerned, we only use~\eqref{eq:ID-upper}, and not~\eqref{eq:ID-lower}} Then,
    \[ 
        \lim_{N\to\infty} \widehat{Q}_N(y,x) = \Sc(\Q|\P)
    \]
    for almost every independent $x \sim \P$ and $y \sim \Q$.\footnote{Equivalently, the statement concerns almost every pair $(x,y)$ with respect to the product measure~$\P\otimes\Q$ on~$\Omega\times\Omega$.}
\end{theorem}

Let us now provide the structure of the proof in the case where $\supp \Q \subseteq \supp \P$, postponing the more technical aspects to Sections~\ref{sec:aux-parse} and~\ref{sec:c-en}.
Throughout, 
\[ 
    \ell_- := \frac{\ln N}{-2\gamma_-}
\] 
and
\[ 
    \ell_+ := \frac{2\ln N}{-\gamma_+}.
\] 
These will serve as \textit{a priori} bounds on the lengths of the words in different auxiliary parsings. 
    \begin{description}
        \item[Upper bound.] 
        Let $\epsilon \in (0, \tfrac 12)$ and $y \in \supp \Q$ be arbitrary.
        We consider an auxiliary sequential parsing $y_1^N= y^{(1,N)} y^{(2,N)} \dotsc y^{(\widehat{c}_N,N)}$, where each word $y^{(j,N)}$ has length $\ell_{j,N}$ and is the shortest prefix of $T^{\ell_{1,N}+\dots+\ell_{j-1,N}}y_1^N$ satisfying 
        \begin{equation}
        \label{eq:ub-aux-pars-req}
            \P[y^{(j,N)}] \leq N^{-1+\epsilon},
        \end{equation} 
        where we define $\ell_{0,N}:=0$. 
        {The power is chosen in the hope that the words in this auxiliary parsing will be long enough, yet likely enough for~$\P$ that the vast majority of them find a match in $x_1^N$. To motivate this Ansatz, note that, by linearity of expectation, the expected number of times a given string of $\P$-probability $N^{-1+\epsilon}$ appears in a string of length~$N$ obtained from~$\P$ grows as $N^{\epsilon}$.}
        
        For $N$ large enough, each length $\ell_{j,N}$ is between~$\ell_-$ and~$\ell_+$, due to Properties \ref{it:FE} and \ref{it:SE}, except possibly for $\ell_{\widehat{c}_N,N}$ which need only satisfy the upper bound. In particular, $\widehat{c}_N = O(\frac{N}{\ln N})$.

        Note that, for each $j = 1,2, \dotsc, \widehat c_N$, the appearance of $y^{(j,N)}$ as a substring of $x_1^N$\,---\,written $y^{(j,N)} \in x_1^N$ in what follows\,---\,implies the presence of at most one separator of the original ZM parsing within $y^{(j,N)}$, that is 
        \[
          \P\left\{ x : \#\{ j \leq \widehat{c}_N : y^{(j,N)} \in x_1^N\} = \widehat{c}_N \right\}\leq \mathbb{P}\{x\,:\, c_N(y|x) \leq  \widehat c_N\},
        \]
        which in turn implies
        \[
            \mathbb{P}\{x\,:\, c_N(y|x) > \widehat c_N\}\leq 
            \P\left\{ x : \#\{ j \leq \widehat{c}_N : y^{(j,N)} \notin x_1^N\} > 0 \right\}.
        \]
        We show in Lemma~\ref{lem:key-prob-ub} that the probability on the right-hand side is summable in~$N$ and hence
        \begin{equation}
        \label{eq:c-leq-c-hat}
            \sum_{N=1}^\infty\P\left\{ x : c_N(y|x) > \widehat c_N \right\}<\infty.
        \end{equation}

        On the other hand, Lemma~\ref{lem:random-Birk-sum} below shows that 
        \begin{equation*}
        \label{eq:c-hat-leq-sum}
            (-1+\epsilon) (\widehat{c}_N-1)  \ln N
                = \sum_{j=1}^{\widehat{c}_N-1} \ln N^{-1+\epsilon} 
                \geq  \sum_{j=1}^{\widehat{c}_N} \ln \P[y^{(j,N)}]
                \geq  \ln \P[y_1^N] - o(N).
        \end{equation*}
        Hence,
        \begin{multline*}
        \mathbb{P}\left\{x\,:\, c_N(y|x)\ln N+\ln \mathbb{P}[y_1^N]>-\frac\epsilon{1-\epsilon}\ln \mathbb{P}[y_1^N]+\ln N + \epsilon N
            \right\} \\
            \leq
            \mathbb{P}\left\{x\,:\, c_N(y|x)\ln N>\widehat c_N\ln N \right\}
        \end{multline*}
        for all~$N$ large enough. Recall that, by Condition \ref{it:SE}, $\ln\P[y_1^N]\geq \gamma_- N$ with $\gamma_- < 0$.
        Combining this with~\eqref{eq:c-leq-c-hat}, we obtain
        \begin{align*}
            &\sum_{N=1}^\infty\mathbb{P}\left\{x\,:\, \frac{c_N(y|x)\ln N}N+\frac{\ln \mathbb{P}[y_1^N]}N>  -\frac{\epsilon}{1-\epsilon}\gamma_- + 2\epsilon
            \right\}<\infty.
        \end{align*}
        Appealing to the Borel--Cantelli lemma, using the cross entropy analogue of the Shannon--McMillan--Breiman theorem in Lemma~\ref{lem:cross-SMB}, and then taking $\epsilon\to 0$, we conclude that, for every $y\in\supp\Q$, we have
        \begin{align*}
            \limsup_{N\to\infty} \widehat{Q}_N(y,x)
            \leq {h^{\mathrm{c}}(\Q|\P)}
        \end{align*}
        for almost every $x\sim\P$.
        
        \item[Lower bound I.] Before we obtain the almost sure lower bound required for Theorem~\ref{thm:simplified}, let us summarize Ziv and Merhav's argument for proving that the lower bound holds in probability. This argument is not logically necessary for the rest of the paper, but may help the reader understand some ideas used for the almost sure version.
        Let $\epsilon \in (0, \tfrac 12)$ and $y \in \supp \Q$ be arbitrary.
        We consider an analogous auxiliary sequential parsing $y_1^N= y^{(1,N)} y^{(2,N)} \dotsc y^{(\bar{c}_N,N)}$, where each word $y^{(j,N)}$ has length $\ell_{j,N}$ and is the shortest prefix of $T^{\ell_{1,N}+\dots+\ell_{j-1,N}}y_1^N$ that has probability 
        \begin{equation}
        \label{eq:lbi-aux-pars-req}
            \P[y^{(j,N)}] \leq N^{-1-\epsilon},
        \end{equation}
        where we define $\ell_{0,N}:=0$. {The power is chosen in the hope that the words in this auxiliary parsing will be numerous enough, yet unlikely enough for~$\P$ that the vast majority of them find no match in $x_1^N$. To motivate this Ansatz, note that the expected number of times a given string of $\P$-probability $N^{-1-\epsilon}$ appears in a string of length~$N$ obtained from~$\P$ decays as $N^{-\epsilon}$.}
            
        Again, for $N$ large enough, each length $\ell_{j,N}$ in this parsing falls between~$\ell_-$ and~$\ell_+$, due to Properties \ref{it:FE} and \ref{it:SE}, except possibly for the last one, which only satisfies the upper bound. In particular, $\bar{c}_N = O(\frac{N}{\ln N})$.
        
        The correspondence between the parsing cardinalities $c_N(y|x)$ and $\bar c_N$ relies on the following observation: $c_N(y|x)$ must be at least equal to the number of words in the auxiliary parsing of $y_1^N$ that \emph{do not} appear as strings in $x_1^N$. Indeed, if a word $y^{(j,N)}$ does not appear as a substring of $x_1^N$\,---\,written $y^{(j,N)} \notin x_1^N$ in what follows\,---, then the ZM parsing has at least one separator within $y^{(j,N)}$. That is
        \[
            c_N(y|x)\geq \#\{j\,:\,y^{(j,N)}\notin x_1^N\}\geq \#\{j\leq \bar c_N-1\,:\,y^{(j,N)}\notin x_1^N\}
        \]
        and so
        \begin{align*}
            \P\left\{x\,:\, c_N(y|x)\geq(\bar c_N-1)\left(1-\epsilon\right)\right\}
            &\geq\P\left\{x\,:\, \#\{j\leq \bar c_N-1\,:\,y^{(j,N)}\in x_1^N\}\leq(\bar c_N-1)\epsilon\right\} \\
            &\geq 1-\P\left\{x: \#\{j\leq \bar c_N-1\,:\,y^{(j,N)}\in x_1^N\}> (\Bar{c}_N-1)\, \epsilon\right\}.
        \end{align*}
        One can easily show using a crude union bound and Markov's inequality that
        the appearance in~$x_1^N$ of more than an arbitrarily small proportion of all the words in the auxiliary parsing except for the last one has vanishing\,---\,but not necessarily summable\,---\,probability, and this enables us to conclude that 
        \begin{equation}
             \lim_{N\to\infty}\P\left\{x : c_N(y|x)\geq(\bar c_N-1)(1-\epsilon)\right\}=1.\label{eq:LBprob}
        \end{equation}

        Note that since, by construction, for any $j = 1,2, \dotsc, \bar c_N$ the auxiliary word $y^{(j,N)}$ has \emph{no} strict prefix with probability less than $N^{-1-\epsilon}$, the lower bound in Condition \ref{it:ID} implies that 
        $
            \P[y^{(j,N)}] \geq N^{-1-2 \epsilon}
        $
        for $N$ large enough.
        Therefore, Lemma~\ref{lem:random-Birk-sum} yields
        \begin{align}
        \label{eq:cbar-geq-sum}
            (-1-2\epsilon) \bar{c}_N \ln N-o(N)
                \leq  \sum_{j=1}^{\bar{c}_N} \ln \P[y^{(j,N)}]  \leq  \ln \P[y_1^N] + o(N), 
        \end{align}
        which together with~\eqref{eq:LBprob}  implies
        \[
            \lim_{N\to \infty}\P\left\{x : \frac{c_N(y|x)\ln N}N +\frac{\ln\P[y_1^N]}N\geq -2\epsilon\frac{\bar{c}_N\ln N}{N}-\epsilon\right\}=1,
        \]

        Thus, using Lemma~\ref{lem:cross-SMB}, the fact that $\bar{c}_N=O(\frac{N}{\ln N})$ and taking $\epsilon\to 0$, we conclude that,
        for all~$y\in\supp\Q$, we have
        \begin{equation}
        \label{eq:LBI-conc}
            h^{\mathrm{c}}(\Q|\P)
            \leq \liminf_{N\to\infty} \widehat{Q}_N(y,x)
        \end{equation}
        in probability with respect to~$x \sim \P$.
       
        \item[Lower bound II.]
        Let $\epsilon \in (0, \tfrac 12)$ and $y \in \supp \Q$ be arbitrary, and fix $0 < \alpha<\frac{\gamma_+}{8\gamma_-} < 1$. In what follows and in the last part of Section~\ref{sec:aux-parse}, the number $N^{\alpha}$ is to be understood as its integer part $\lfloor N^{\alpha}\rfloor$. We modify the auxiliary parsing of $y_1^{N}$ as follows. 
        
        First, let $y^{(1,1,N)}$ be the shortest prefix of $y_1^{N^{\alpha}}$ such that $\P[y^{(1,1,N)}] \leq N^{-1-\epsilon}$; it has length $\ell_{1,1,N}$, between~$\ell_-$ and~$\ell_+$ for $N$ large enough due to Properties \ref{it:FE} and \ref{it:SE}. Now, let $y^{(2,1,N)}$ be the  shortest prefix of $y_{\ell_{1,1,N}+1}^{N^{\alpha}}$ such that $\P[y^{(2,1,N)}] \leq N^{-1-\epsilon}$, and so on until not possible. We have parsed a first \emph{block} of size $N^\alpha$:
        \begin{equation*}
            y_1^{N^\alpha} = y^{(1,1,N)} y^{(2,1,N)} \dotsb y^{(d_{1,N},1,N)} \xi^{(1,N)}
        \end{equation*}
        where the (possibly empty) \emph{buffer} $\xi^{(1,N)}$ has probability at least $N^{-1-\epsilon}$ and length at most~$\ell_+$ due to Property~\ref{it:FE}. 
        
        We then repeat the procedure with $T^{N^\alpha}y_1^N$ to obtain the second block, and so on until 
        \begin{multline}
        \label{eq:block-parsing}
            y_1^{N} 
                = y^{(1,1,N)} y^{(2,1,N)} \dotsb y^{(d_{1,N},1,N)} \xi^{(1,N)} y^{(1,2,N)} y^{(2,2,N)} \dotsb y^{(d_{2,N},2,N)} \xi^{(2,N)} \\
                \dotsb
                y^{(1,M_N,N)} y^{(2,M_N,N)} \dotsb y^{(d_{M_N,N},M_N,N)} \xi^{(M_N,N)}.
        \end{multline}
        {The construction of $y^{(1,M_N,N)} y^{(2,M_N,N)} \dotsb y^{(d_{M_N,N},M_N,N)} \xi^{(M_N,N)}$ may differ from that of $y^{(1,s,N)} y^{(2,s,N)} \dotsb y^{(d_{s,N},s,N)} \xi^{(s,N)}$ for $s < M_N$ in that it might be the parsing of a block of a length smaller than $N^\alpha$ if there is a remainder in the division of~$N$ by $N^\alpha$.}
        Note that, for $N$ large enough, $N^{1-\alpha}\leq M_N \leq 2N^{1-\alpha}$ and $d_{s,N}\leq \frac{2N^\alpha}{\ell_-}\eqqcolon d_+$. The number of auxiliary parsed words to be considered is
        \[ 
            \Tilde{c}_N := d_{1,N} + d_{2,N} + \dotsb +d_{M_N,N}.
        \]
        It follows from the above that $\Tilde{c}_N=O(\frac{N}{\ln N})$, since $ d_{s,N}\geq\frac{N^{\alpha}}{2\ell_+}\eqqcolon d_-$ for any $s<M_N$.
        As explained in ``Lower bound~I'', $c_N(y|x)$ must be at least equal to the number of words in the auxiliary parsing of $y_1^N$ that \emph{do not} appear as strings in $x_1^N$.
        In order to control the latter, we prove below the two following technical estimates:
        \begin{itemize}
            \item Proposition~\ref{prop:c-mostly-good}: For  almost every~$y\sim\Q$, there exists $N_\epsilon(y)$ such that, for $N \geq N_\epsilon(y)$, the number of indices $s$ such that the words $y^{(1,s,N)}$,  $y^{(2,s,N)}, \dotsc, y^{(d_{s,N},s,N)}$ are \emph{not} distinct is smaller than $\epsilon M_N$.
            
            \item Proposition~\ref{prop:c-limited-in-good}: Denoting by $\mathcal{S}_{\mathrm{g}}(y_1^N)$ the set of indices~$s$ whose block of~$y_1^N$ \emph{does} consist of distinct words,
            we have
            \begin{align*}
                \P\{\#\{j : y^{(j,s,N)} \in x_1^N\} > \epsilon d_+\} \leq \ell_+^2 \Exp{\frac{\epsilon^2 \gamma_+ }{8}\frac{d_+}{\ell_+}}
            \end{align*}
            for $N$ large enough and all~$s \in \mathcal{S}_{\mathrm{g}}(y_1^N)$. This means that, with high probability, only a small fraction of the words in these ``good blocks'' can appear in~$x_1^N$ (and fail to contribute to~$c_N(y|x)$).
    \end{itemize}
        
    Therefore, even considering the worst-case scenario where all $y^{(i,s,N)}$ with $s\notin \mathcal{S}_{\mathrm{g}}(y_1^N)$ do appear in~$x_1^N$, we find that, for  almost every~$y\sim\Q$,
    \begin{align*}
        &\sum_{N = N_\epsilon(y)}^\infty \P\left\{x\,:\, c_N(y|x) <\Tilde{c}_N - 2\epsilon d_+M_N  \right\}\\
        &\hspace{1cm}\leq \sum_{N = N_\epsilon(y)}^\infty \P\left\{x\,:\, \sum_{s = 1}^{M_N} \#\left\{j : y^{(j,s,N)} \notin x_1^N\right\}< \Tilde c_N-2\epsilon d_+ M_N\right\} \\
        &\hspace{1cm}= \sum_{N = N_\epsilon(y)}^\infty \P\left\{x\,:\, \sum_{s = 1}^{M_N} \#\left\{j : y^{(j,s,N)} \in x_1^N\right\} > 2\epsilon d_+ M_N\right\} \\
        &\hspace{1cm}\leq \sum_{N = N_\epsilon(y)}^\infty \P\left\{x\,:\, \sum_{s \in \mathcal{S}_{\mathrm{g}}(y_1^N)} \#\left\{j : y^{(j,s,N)} \in x_1^N\right\}> \epsilon d_+ M_N\right\} \\
        &\hspace{1cm}\leq \sum_{N = N_\epsilon(y)}^\infty M_N \max_{s \in \mathcal{S}_{\mathrm{g}}(y_1^N)}\P\{\#\{j : y^{(j,s,N)} \in x_1^N\} > \epsilon d_+\}
    \end{align*}
    is finite.
    
    Appealing to Lemma~\ref{lem:random-Birk-sum} and Remark~\ref{rem:CEmod}, the relation \eqref{eq:cbar-geq-sum} between $\Tilde{c}_N$ and $\ln \P[y_1^N]$ remains valid and yields
    \begin{multline*}
        \P\left\{x: \frac{c_N(y|x)\ln N}N +\frac{\ln \P[y_1^N]}N<-2\epsilon \frac{\Tilde{c}_N\ln N}N-\epsilon-8\epsilon\frac{ \ln N}{\ell_-} \right\}\\
        \leq \P\left\{x: \frac{c_N(y|x)\ln N}N <\frac{\Tilde{c}_N\ln N}N - 2\epsilon d_+M_N\frac{\ln N}N \right\}
    \end{multline*}
    for $N$ large enough, which implies that there exists some constant $C=C(\gamma_-)>0$ such that
     \[
        \sum_{N=1}^\infty\P\left\{x: \frac{c_N(y|x)\ln N}N +\frac{\ln \P[y_1^N]}N< -C\epsilon \right\}<\infty.
    \]
    By Lemma~\ref{lem:cross-SMB} and the Borel--Cantelli lemma, taking $\epsilon\to 0$ we conclude that 
    \begin{align*}
        h^{\mathrm{c}}(\Q|\P) \leq \limsup_{N\to\infty} \widehat{Q}_N(y,x)
    \end{align*}
    for almost every independent $x\sim\P$ and $y\sim\Q$.
\end{description}

The above strategy is essentially that of Ziv and Merhav, but the lemmas and propositions on which it relies need to be adapted beyond Markovianity. Before we do so, let us state and prove a proposition that justifies our focus on situations where $\supp \Q \subseteq \supp \P$.

\begin{proposition}
    Suppose that $\Q$ is ergodic. If there exists $k\in\nn$ such that $\supp \Q_k \cap \supp \P_k^\mathsf{C} \neq \emptyset$, then $\widehat{Q}_N\to\infty$ almost surely as $N\to\infty$, in agreement with Theorem~\ref{thm:simplified}.
\end{proposition}

\begin{proof}
    Fix $k$ as in the hypothesis and then $a \in \supp \Q_k \setminus \supp \P_k$. Because $a \notin \supp \P_k$, a crude counting argument yields that the ZM parsing satisfies
    \[ 
        c_N(y|x) \geq \frac{\#\{j \leq N-k+1 :  T^{j-1}y \in [a]\}}{k}
    \]
    for all $x \in \supp \mathbb{P}$.
    Because $a \in \supp \Q_k$ and $k$ is fixed, Birkhoff's ergodic theorem applied to the function~$\mathbf{1}_{[a]}$ yields
    \[ 
        \liminf_{N\to\infty} \frac{c_N(y|x)}{N} > 0,
    \]
    for  almost every~$y\sim\Q$.
    This allows us to conclude that, almost surely, the estimator diverges.

    As for the claim that this is in agreement with Theorem~\ref{thm:simplified}, it is based on the observation that if $\supp \Q_k \cap \supp \P_k^\mathsf{C} \neq \emptyset$, then $\supp \Q_{n} \cap \supp \P_{n}^\mathsf{C} \neq \emptyset$ for all~$n \geq k$. Since the existence of $a \in \supp \Q_n$ such that $\P_n[a] = 0$ causes at least one summand to be infinite on the right-hand side of~\eqref{eq:def-cross}, this allows us to conclude that $\Sc(\Q|\P) = \infty$.
\end{proof}

\subsection{Properties of the auxiliary parsings}
\label{sec:aux-parse}

Throughout this section, $\epsilon\in(0,1/2)$ is fixed but arbitrary. We assume that $\P$ and $\Q$ are stationary
and satisfy $\supp\Q\subseteq\supp\P$. For readability, we will omit keeping track of the $N$-dependence in
some of the notation introduced above.
As foreshadowed in the introduction, our analysis of the cardinalities of the auxiliary parsings will use reformulations in terms of waiting times. 

\begin{lemma}
\label{lem:key-prob-ub}
    Suppose that $\P$ satisfies \ref{it:ID}, \ref{it:FE}, and \ref{it:KB}. Let $y \in \supp \Q$ be arbitrary and consider the auxiliary parsing of~$y_1^N$ built around the requirement~\eqref{eq:ub-aux-pars-req}. Then,
    \[ 
        \P\left\{ x : \#\{ j \leq \widehat{c}_N : y^{(j)} \notin x_1^N\} > 0 \right\} \leq N\Exp{-\frac{N^{\frac{\epsilon}{4}}}{3\ell_+}}
    \]
    for $N$ large enough.
    
\end{lemma}

\begin{proof}
    Let $\underline{y}^{(j)}$ be the word that is obtained by removing the last letter from~$y^{(j)}$; by construction, $\P[\underline{y}^{(j)}] > N^{-1+\epsilon}$. So, in view of \ref{it:ID},
    \begin{equation}
    \label{eq:use-1-of-Ad}
        \P[y^{(j)}] \geq \Exp{-k_{\ell_j-1}} \P[\underline{y}^{(j)}] \left(\min_{a \in \mathcal{L}_1} \P[a]\right) \geq N^{-1+\frac \epsilon 2}
    \end{equation}
    for $N$ large enough. We have used the fact that $k_{\ell_j-1} \leq k_{\ell_+}$ with $k_\ell = o(\ell)$ and $\ell_+=O(\ln N)$. Using \ref{it:KB} and considering all $N$ large enough, we have 
    \begin{align}
    \label{eq:use-1-of-Ad-then}
        \P\{ x\,:\,W_{\ell_j}(y^{(j)},x)>N-\ell_j+1 \}
            &\leq \exp\left(-\frac{N^{\frac{\epsilon}{4}}}{2\ell_++\tau_{\ell_+}}\right),
    \end{align}
    where we used 
    the defining properties of $k_{\ell_j}$ and $\ell_j$. Then, using that for $N$ large enough we have $\tau_{\ell_+}\leq\ell_+$ and taking a union bound over $j$ 
    \[ 
        \P\left(\bigcup_{j=1}^{\widehat c_N}\{x : W_{\ell_j}(y^{(j)},x)>N-\ell_j+1\} \right) \leq N\exp\left(-\frac{N^{\frac{\epsilon}{4}}}{3\ell_+}\right). 
    \]
    To conclude, note that $W_{\ell_j}(y^{(j)},x)>N-\ell_j+1$ is a necessary and sufficient condition for~$y^{(j)} \notin x_1^N$.
\end{proof}

While, on one hand, the last lemma states that the words in the auxiliary parsing built around~\eqref{eq:ub-aux-pars-req} tend to appear in~$x_1^N$, one can show that, on the other hand, the words in the auxiliary parsing built around~\eqref{eq:lbi-aux-pars-req} tend to \emph{not} appear in~$x_1^N$. However, the probabilistic estimate obtained pursuing this strategy only achieves convergence in probability of the ZM estimator.
%
%
%
As Ziv and Merhav showed in their original paper in the Markovian case, this estimate can actually be refined and made summable in~$N$ using some additional combinatorial and probabilistic arguments. Such a refinement is used to go from convergence in probability to almost sure convergence in Section~\ref{ssec:proof-outline}.
We recall the following basic facts about our modified auxiliary parsing~\eqref{eq:block-parsing} for $N$ large enough:
\begin{itemize}
    \item {there are $M_N \leq 2N^{1-\alpha}$ blocks, indexed by $s$, each of length $N^\alpha$ except for the last one  $(s=M_N)$ which possibly has length less than $N^\alpha$};
    \item the $s$-th block contains $d_{s}$ words $y^{(i,s)}$ with 
    \begin{align*}
        d_{-}\coloneqq\frac{N^{\alpha}}{2\ell_+} \leq d_{s} \leq\frac{2N^{\alpha}}{\ell_-} \eqqcolon d_{+},
    \end{align*}
    {except for the last one  $(s=M_N)$ for which the lower bound may not apply}, and one (possibly empty) buffer $\xi^{(s)}$; 
    \item each word $y^{(i,s)}$ has length $\ell_{i,s}$, with
    \[ 
        \ell_- := \frac{\ln N}{-2\gamma_-} \leq \ell_{i,s} \leq  \frac{2\ln N}{-\gamma_+} =: \ell_+.
    \]
\end{itemize}
Most of the factors of 2 in these facts are suboptimal; they are only meant to avoid having to consider integer parts or superficial dependence on~$\epsilon$.

\begin{definition}
    If the words $y^{(1,s)}$, $y^{(2,s)}, \dotsc, y^{(d_{s},s)}$ in~\eqref{eq:block-parsing} are all distinct, we say that the $s$-th block of $y_1^N$ is \emph{good} and write $s \in \mathcal{S}_{\mathrm{g}}(y_1^N)$. If that is not the case, we say that the block is \emph{bad} and write $s \in \mathcal{S}_{\mathrm{b}}(y_1^N)$.
\end{definition}

\begin{lemma}\label{lem:A1}
    If $\Q$ satisfies \ref{it:ID} and \ref{it:FE}, then
    \[
       \Q\{y : s \in \mathcal{S}_{\mathrm{b}}(y_1^N)\} \leq \Exp{k_{\ell_-}}N^{-2\alpha},
    \]
    for every~$s$ and every $N$ large enough. 
\end{lemma}

\begin{proof} 
    Fix $\Q$ as in the statement.
    By shift-invariance, $\Q\{y : s \in \mathcal{S}_{\mathrm{b}}(y_1^N) \} \leq \Q\{y : 1 \in \mathcal{S}_{\mathrm{b}}(y_1^N)\}$.\footnote{In fact, as long as $s < M_N$, the probabilities are equal. }
    For the first block to be bad, two words $y^{(i,1)}$ and $y^{(j,1)}$  need to coincide, and in particular, their $\ell_-$-prefixes need to coincide. Hence, considering all possible starting indices of these two words, and appealing to shift-invariance, \ref{it:ID} and \ref{it:FE}, we derive
    \begin{align*}
       \Q\{y : 1 \in \mathcal{S}_{\mathrm{b}}(y_1^N)\}
            &\leq \sum_{r = \ell_-}^{N^{\alpha}} \sum_{r' = 0}^{r-\ell_-} \sum_{u \in \supp \Q_{\ell_-}} \Q(T^{-r'}[u] \cap T^{-r}[u])\\
            &\leq \binom{N^{\alpha}}{2}\sum_{u \in \supp \Q_{\ell_-}} \Exp{k_{\ell_-}} \Q[u]^2\\
            &\leq N^{2\alpha}\Exp{k_{\ell_-}}\Exp{\gamma_+ \ell_-}.
    \end{align*}   
    To conclude, recall that we have chosen $\alpha<\frac{\gamma_+}{8\gamma_-} = - \frac{\gamma_+\ell_-}{4\ln N}$ and that $\gamma_\pm < 0$. 
\end{proof}

\begin{lemma}
\label{lem:m-bad-blocks}
    If $\Q$ satisfies \ref{it:ID} and \ref{it:FE}, then 
    \[
        \Q\{y : \#\mathcal{S}_{\mathrm{b}}(y_1^N) = m \}\leq \binom{M_N}{m}\Exp{2mk_{\ell_-}}N^{-2m\alpha},
    \]
   for all~$m\in\nn$ and for all $N$ large enough.
\end{lemma}

\begin{proof}
    Fix $\Q$ and $m$ as in the statement. Let us first consider the probability that the blocks of $y_1^N$ labeled $s_m$, $s_{m-1}$ down to $s_1$ are bad. This event can be thought of as $m$-th in a sequence of events defined inductively by 
    $E'_{k+1}=T^{-N^{\alpha}(s_{k+1}-1)}\{1 \in \mathcal{S}_{\mathrm{b}}\} \cap E'_k$
    where $E'_0=\Omega$. It follows, by a straightforward adaptation of the strategy of Lemma~\ref{lem:A1}, that 
    \[
        \Q(E'_{k+1})\leq \binom{N^{\alpha}}{2}\sum_{u \in \supp \Q_{\ell_-}}\Exp{2k_{\ell_-}}\Q[u]^2\Q(E'_k)\leq N^{2\alpha}\Exp{2k_{\ell_-}}\max_{u \in \supp \Q_{\ell_-}}\Q[u]\Q(E'_k). 
    \]
    Iterating and accounting for the different choices of $s_1, \dotsc, s_{m-1}, s_m$ (recall that $s \leq M_N$) gives the proposed bound.
\end{proof}

\begin{proposition}
\label{prop:c-mostly-good}
   If $\Q$ satisfies~\ref{it:ID} and \ref{it:FE}, then  
   for  almost every~$y\sim\Q$, there exists $N_\epsilon$ such that $\#\mathcal{S}_{\mathrm{b}}(y_1^N) < \epsilon M_N$ for all $N \geq N_\epsilon$.
\end{proposition}

\begin{proof}
     Fixing $\Q$ as in the statement, using Markov's inequality, the binomial theorem and Lemma~\ref{lem:m-bad-blocks}, for every $b>0$ we have
    \begin{align*}
         \Q\left\{y: \#\mathcal{S}_{\mathrm{b}}(y_1^N)  \geq \epsilon M_N\right\}&\leq \mathbb{E}\left(\Exp{b(\#\mathcal{S}_{\mathrm{b}}(y_1^N)) }\right)\Exp{-b\epsilon M_N}\\
         &=\Exp{-b\epsilon M_N}\sum_{m=1}^{M_N}\Exp{bm}\Q\left\{y: \#\mathcal{S}_{\mathrm{b}}(y_1^N) = m \right\}\\
         &\leq\Exp{-b\epsilon M_N}\left(1+\Exp{b+2k_{\ell_-}}N^{-2\alpha}\right)^{M_N}.
    \end{align*}
    Choosing $b=2\alpha\ln N-2k_{\ell_-}$, recalling that $M_N/N^{1-\alpha}\in (1,2)$ and considering $N$ large enough so that $b>0$ gives the bound 
    \begin{align}
       \label{eq:Q-bound}
       \Q\left\{y: \#\mathcal{S}_{\mathrm{b}}(y_1^N) \geq \epsilon M_N\right\}\leq \Exp{-N^{1-\alpha}(2\alpha\epsilon\ln N-2\epsilon k_{\ell_-}-2\ln 2)}.
    \end{align}
    The proposition thus follows from the Borel--Cantelli lemma.
\end{proof}

\begin{lemma}
\label{lem:m-words-in-lk}
    Suppose that $\P$ satisfies \ref{it:ID} and that the $s$-th block of $y_1^N$ is good. 
    Given $\ell$ and $K\in\{1,2\dotsc, \ell\}$,
    \begin{multline*}
        \P\left\{x:\#\left\{j :
        y^{(j,s)}=x_{K+r\ell}^{K+r\ell+(\ell-1)}\text{ for some } 
        r \in \left\{0, 1, \dots,\left\lfloor\frac{N-K+1}{\ell}\right\rfloor-1\right\}
        \right\}=m\right\}
        \\ \leq \binom{d_+}{m}\Exp{mk_{\ell_+}}N^{-m\epsilon}.
    \end{multline*}

\end{lemma}

\begin{proof}
    By shift invariance, we can assume that $s=1$.
    Consider a set $I=\{i_k\}_{k=1}^m$ of $m$ distinct indices such that $y^{(i_k,1)}$ has length $\ell$, and let $F(I)$ denote the event that all the words $\{y^{(i_k,1)}\}_{k=1}^m$ have a match in $x_1^N$ with a starting point equivalent to $K$ mod~$\ell$. Since the words $\{y^{(i_k,1)}\}_{k=1}^m$ are distinct, the starting positions of the matches considered must be distinct. 
    Moreover, by assumption, each such starting position is of the form~$r\ell + K$ for some~$r$ at most $\lfloor\tfrac{N-K+1}{\ell}\rfloor - 1$. Therefore, enumerating all possibilities, we find
    \begin{align*}
        F(I)
            \subseteq \bigcup_{r_1,\dotsc,r_{m}}\bigcap_{k=1}^{m}T^{-r_k\ell-K }[y^{(i_{k},\ell)}],
    \end{align*}
    where the union is taken over distinct nonnegative integers $r_1, \dotsc, r_m$ all at most $\lfloor\tfrac{N-K+1}{\ell}\rfloor-1$.
    Using \ref{it:ID}, shift-invariance and subadditivity gives
    \[
        \P(F(I))\leq m!\binom{\lfloor \frac{N-K+1}{\ell}\rfloor-1}{m}\left(\Exp{k_{\ell_+}}\max_{i \in I}\P[y^{(i,\ell)}]\right)^m\leq N^m(\Exp{k_{\ell_+}}N^{-1-\epsilon})^m\leq \Exp{mk_{\ell_+}}N^{-m\epsilon}.
    \]
    To conclude, we use a union bound, together with an upper bound on the number of sets~$I$ of this nature.
\end{proof} 

\begin{remark} 
    The separation into fixed values of~$\ell$ and~$K$ is a technical device to avoid overlaps that would prevent the use of~\ref{it:ID}, and will be taken care of momentarily by a union bound.
    For fixed~$\ell$, and
    for the purpose of relating $c_N$ and $\Tilde{c}_N$, the important quantity is the number of~$j$ such that $y^{(j,s)}$ has size~$\ell$ and appears in $x_1^N$ (this is the only way a separator could fail to appear within $y^{(j,s)}$), and not the number of substrings of size~$\ell$ in~$x_1^N$ that are matches for some $y^{(j,s)}$. The probability of the latter is easier to control (this is what we control in the proof), and coincides with the former when $s \in \mathcal{S}_{\mathrm{g}}(y_1^N)$. 
\end{remark}

\begin{proposition}  
\label{prop:c-limited-in-good}
    Let $y \in \supp \Q$ be arbitrary and consider the modified auxiliary parsing of~$y_1^N$ in~\eqref{eq:block-parsing}. Suppose that $\P$ satisfies \ref{it:ID} and that the $s$-th block of $y_1^N$ is good. Then, for $N$ large enough, the event that more than a fraction $\epsilon$ of the maximum number $d_+$ of words $y^{(i,s)}$ in the $s$-th block appears in $x_1^N$ satisfies
    \begin{equation}
    \label{eq:P-prob-bound}
       \P\{x:\#\{j : y^{(j,s)} \in x_1^N\} > \epsilon d_+\}\leq \ell_+^2\Exp{ \frac{\gamma_+\epsilon^2}8\frac{d_+}{\ell_+}}.
    \end{equation}
\end{proposition}

\begin{proof}
    Fix~$s\in\mathcal{S}_{\mathrm{g}}(y_1^N)$. Given $\ell$ and $K \in \{1,\dotsc,\ell\}$, consider
    \begin{equation}
        \label{eq:P-prob-bound-fixed-kl}
        \chi_{(K,\ell)} := \sum_{i :  \ell_{i,s} = \ell} \mathbf{1}_{W_\ell(y^{(i,s)},\,\cdot\,)\leq N} \cdot \mathbf{1}_{W_\ell(y^{(i,s)},\,\cdot\,) \equiv_{\operatorname{mod} \ell} K} 
    \end{equation}
    Observe that for any fixed $x$,
    \[
        \#\{j : y^{(j,s)} \in x_1^N\}
        \leq \sum_{i=1} ^{d_{s}}\mathbf{1}_{W_{\ell_{i,s}}(y^{(i,s)},x)\leq N}
    \]
    and so for the random variable in~\eqref{eq:P-prob-bound} to exceed $\epsilon  d_+$, at least one of the random variables $\chi_{(K,\ell)}$ defined by~\eqref{eq:P-prob-bound-fixed-kl} must exceed $\tfrac{\epsilon  d_+}{\ell_+^2}$, that is
    \begin{align}
        \P\{x:\#\{j : y^{(j,s)} \in x_1^N\} > \epsilon d_+\}
        &\leq\P\left( \bigcup_{(K,\ell)} \left\{x: \chi_{(K,\ell)}(x)> \epsilon\frac{d_+}{\ell_+^2 } \right\} \right). \label{eq:UBgood}
    \end{align}
   
    Following the same strategy as in the proof of Proposition~\ref{prop:c-mostly-good}, we use Markov's inequality, the binomial theorem and Lemma~\ref{lem:m-words-in-lk} to derive that, for every $b>0$,
    \[ 
        \P\left\{x: \chi_{(K,\ell)}(x)> \epsilon\frac{d_+}{\ell_+^2} \right\}\leq \left(1+\frac{\Exp{b+k_{\ell_+}}}{N^{\epsilon}}\right)^{d_+}\Exp{-b\epsilon\frac{d_+}{\ell_+^2 }}.
    \]
    Choosing $b=\frac{\epsilon}{2}\ln N$ yields 
    \begin{align*}
        \P\left\{x: \chi_{(K,\ell)}(x)> \epsilon\frac{d_+}{\ell_+^2} \right\}
        &\leq \Exp{ \frac{\gamma_+\epsilon^2}4 \frac{d_+}{\ell_+}(1-o(1))}\leq \Exp{ \frac{\gamma_+\epsilon^2}8\frac{d_+}{\ell_+}}
    \end{align*}
    for $N$ large enough, recalling that $\gamma_+<0$. 
    Going back to our observation \eqref{eq:UBgood},
    we conclude the proof by performing a union bound over~$K$ and~$\ell$.
\end{proof}
\subsection{Cross entropy}
\label{sec:c-en}
\begin{lemma}
\label{lem:random-Birk-sum}
    If $\P$ satisfies \ref{it:ID}, $y\in\supp\P$ and $y_1^N$ is parsed as
    \[ 
        y_1^N = y^{(1,N)} y^{(2,N)} \dotsc y^{(c'_N-1,N)} y^{(c'_N,N)},
    \]
    with $\ell_j\coloneqq |y^{(j,N)}| \geq \lambda_N$ for some properly diverging, nonnegative sequence $(\lambda_N)_{N=1}^\infty$, then 
    \[
        \sum_{j=1}^{c'_N} \ln \P[y^{(j,N)}] = \ln \P[y_1^N] + o(N).
    \]
\end{lemma}

\begin{proof}
    Suppose $\P$ satisfies \ref{it:ID}, $y\in\supp\P$ and $y_1^N$ is parsed as in the statement. Both the upper and lower bound are proved similarly so we only provide the proof of the former. Let $\epsilon > 0$ be arbitrary and note that \ref{it:ID} yields
    \begin{align*}
        \ln \P[y_1^N]&=\ln \P[y^{(1,N)} y^{(2,N)} \dotsc y^{(c'_N-1,N)} y^{(c'_N,N)}]\\
        &\leq \ln \left(\Exp{k_{\ell_1}+\dots+k_{\ell_{c'_N-1}}}\P[y^{(1,N)}]\P[ y^{(2,N)}] \dotsc \P[ y^{(c'_N,N)}]\right)\\
        &=\sum_{j=1}^{c'_N}\ln \P[y^{(j,N)}]+\sum_{j=1}^{c'_N-1}k_{\ell_j}.
    \end{align*}
    Now since $k_\ell = o(\ell)$ and $\lambda_N \to \infty$, we have $k_{\ell_j} < \epsilon \ell_{j}$ for $N$ large enough. Therefore, 
    \begin{align*}
        \ln \P[y_1^N]
            &< \sum_{j=1}^{c'_N}\ln \P[y^{(j,N)}]+\sum_{j=1}^{c'_N-1}\epsilon {\ell_j} \\
            &< \sum_{j=1}^{c'_N}\ln \P[y^{(j,N)}] + \epsilon N
    \end{align*}
    for $N$ large enough.
\end{proof}

\begin{remark}\label{rem:CEmod} Note that the contribution coming from the buffers $\xi^{(s,N)}$, with $s\in\{1,\dotsc,M_N\}$, in the modified auxiliary parsing \eqref{eq:block-parsing} can be embedded in the correction term $o(N)$ in the statement of Lemma~\ref{lem:random-Birk-sum}. This immediately follows by observing that $M_N=o(\Tilde c_N)$.
\end{remark}

\begin{lemma}\label{lem:cross-SMB}
    If $\P$ satisfies \ref{it:ID} and $\Q$ is ergodic, and if $\supp \Q \subseteq \supp \P$, then
    \[ 
        -\ln \P[y_1^N] = N\Sc(\Q|\P) + o(N)
    \]
    for  almost every~$y\sim\Q$.
\end{lemma}

\begin{proof}
    Fix $\P$ and $\Q$ as in the statement. In view of the upper bound in \ref{it:ID}, we can apply Kingman's subadditive ergodic theorem to the sequence $(f_n)_{n=1}^{\infty}$ of measurable functions on the dynamical system~$(\supp \P, T, \Q)$ defined by $f_n(x) := \ln \P[x_1^n]$.
\end{proof}

\subsection{Comments}
\label{ssec:comments}
The following consequence of \ref{it:ID} played an important role in the proof of the upper bound:
\begin{description}
    \item[{Ad}\namedlabel{it:Ad}{Ad}] For every $n \in \nn$, the bound
        \begin{equation}
        \label{eq:add-letter}
            \min \left\{ \frac{\P[ab]}{\P[a]} : a \in \supp \P_n, b \in \supp \P_1, ab \in \supp \P_{n+1} \right\} \geq \Exp{-k_n}
        \end{equation} 
        holds.
\end{description}
Indeed, by construction of Ziv and Merhav's auxiliary parsings, there is a lower bound on $\P[\underline{y}^{(j,N)}]$ and an upper bound on $\P[y^{(j,N)}]$, but both the bounds~\eqref{eq:c-leq-c-hat} and~\eqref{eq:cbar-geq-sum} require a lower bound on $\P[{y^{(j,N)}}]$; see Lemma~\ref{lem:key-prob-ub}.
Condition~\ref{it:Ad} serves as a way of going back and forth between the two. Unfortunately, {\ref{it:Ad}} may fail upon relaxing the lower bound in {\ref{it:ID}} to the more general lower-decoupling conditions that have met with success in tackling other related problems~\cite{CJPS19,BCJPPExamples,CDEJR23w,CR23}. We will come back to this point in Section~\ref{ssec:hmm}.

As for the arguments available in the literature to establish~\ref{it:KB}, we foresee no difficulty in adapting our argument to a set of hypotheses where the roles of~$a$ and~$b$ are exchanged in the decoupling inequalities. Indeed, this would not affect \ref{it:SE} nor~\ref{it:Ad}. While the Markov property can be equivalently written in terms of conditioning on the past or conditioning on the future, the class of g-measures discussed in Section~\ref{sec:examples} and its ``reverse'' counterpart do not coincide; see e.g.~\cite[\S{4.4}]{BFV19}.

As mentioned in the introduction, the Ziv--Merhav estimator can be written in terms of longest-match lengths:
\begin{align*}
    \frac{c_N \ln N}{N}
    &= \frac{\ln N}{\frac{1}{c_N}\sum_{i=1}^{c_N} \ell^{(i,N)}}
\end{align*}
where%
    \footnote{The minimum over the two terms will be given by the former as long as $i<c_N$. However, this formulation is necessary to take care of the ``edge cases'' alluded to in the Introduction.}
\[
    \ell^{(i,N)} = \min\{\Lambda_N(T^{L^{(i-1,N)}}y,x), N-L^{(i-1,N)}\},
\]
with
\begin{align*}
    L^{(0,N)} = 0
    \quad&\text{and}\quad
    L^{(i,N)} = L^{(i-1,N)} + \ell^{(i,N)}
\end{align*}
for $i = 1,2,\dotsc, c_N$. It is known that the longest-match estimator~$(\ell^{(1,N)})^{-1}\ln N = \Lambda_N(y,x)^{-1} \ln N$ converges almost surely to the cross entropy, with good probability estimates, for a class of measures that is more general than that considered here; see~\cite[\S{1.3}]{Ko98} and~\cite[\S{3}]{CDEJR23w}. Hence, if each $T^{L^{(i-1,N)}}y$ were replaced by a new independent sample from~$\Q$, or by $T^{\Delta(i - 1)}y$ for some fixed deterministic $\Delta \in \nn$, then one would expect the convergence of the Ziv--Merhav estimator to also hold considerably more generally. However, the dependence structure of the starting indices seems to be posing a serious technical difficulty for the strategy of Ziv and Merhav. 

\section{Examples}
\label{sec:examples}


{In this section, we discuss broad classes of measures to which our results apply. For this discussion, we need basic topological considerations that we had avoided so far. A one-sided (resp.\ two-sided) \emph{subshift} is a closed subset of~$\cA^\nn$ (resp.~$\cA^\zz$) obtained by removing all sequences containing at least one string from some set of \emph{forbidden strings}. Closure is understood in the product topology, and the subshift is equipped with the subspace topology inherited from that topology. A subshift is said to be \emph{of finite type} if the list of forbidden strings that defines it can be chosen to be finite. A subshift of finite type is said to be \emph{topologically transitive} if, for any two strings $a$ and $b$ with $[a]$ and $[b]$ intersecting the subshift, there exists a third string~$\xi$ such that $[a\xi b]$ also intersects the subshift. We refer the reader to~\cite[\S{7}]{DeGrSi} or~\cite[\S{8}]{KLO} for a more thorough discussion.} 

\subsection{Markov measures}

As mentioned in Section~\ref{sec:setting}, if $\P$ is the stationary measure for an irreducible Markov chain with positive entropy, then $\P$ is ergodic and satisfies \ref{it:ID}, \ref{it:FE}, and \ref{it:KB}. We use this setting to illustrate the role of some of our conditions. 

Note that in the case of a \emph{reducible} Markov chain, a stationary measure can charge two disjoint communication classes; let us call those classes~$\cA'$ and $\cA''$. Then, for $a \in \cA'$, the probability $\P\{x : W_1(a,x) \geq r\} \geq \P\{x : x_1 \in \cA''\}$ does not decay as $r \to \infty$. 
In terms of the language of subshifts, the failure of~\ref{it:KB} is due to the fact that~$\supp\P$ does not satisfy any form of specification; it is a subshift of finite type that fails to be transitive. More concretely, if the sequence~$x$ starts in~$\cA''$, then it remains there forever and we do not expect to be able to probe any entropic quantity that also involves the behaviour of~$\P$ on~$\cA'$ using the information contained in~$x$.

Also note that a stationary measure for an irreducible Markov chain could fail to have positive entropy if, for example, it is a convex combination of Dirac masses on periodic orbits. 
Such a behaviour is at odds with \ref{it:FE} and can cause the bounds on the lengths of the parsed words not to be controlled in terms of $\ell_\pm$, a fact which was used repeatedly throughout our proofs.

\newcommand{\rgm}{regular g-measure}
\newcommand{\Rgm}{Regular g-measure}

\subsection{\Rgm{s}}

Let $\Omega'$ be a topologically transitive one-sided subshift of finite type. Choosing as a starting point one particular definition in the literature among others, we will say that a translation-invariant measure~$\P$ on~$\Omega$ is a \emph{\rgm{} on~$\Omega'$} if $\supp \P = \Omega'$ and there exists a continuous function~$g: \Omega' \to (0,1]$ such that
\begin{equation}\label{eq:g1}
    \sum_{\substack{y \in \Omega'\\ Ty = x}} g(y) = 1
\end{equation}
for all~$x\in\Omega'$ and
\begin{equation} 
\label{eq:gn-to-g}
    \lim_{n\to\infty} \sup_{x \in \Omega'} \left|\frac{\P[x_1^{n}]}{\P[x_2^n]} - g(x)\right| = 0.
\end{equation}
The convergence~\eqref{eq:gn-to-g}
can be used to show that $\P$ satisfies the decoupling condition~\ref{it:ID}; see~\cite[\S{B.3}]{CR23}. Our assumption on~$\Omega'$ more than suffices for~\ref{it:ID} to yield~\ref{it:KB}; see~\cite[\S\S{3.1,\,B.2}]{CR23}

Note that the ratio being compared to $g$ is continuous in~$x$ at finite~$n$, and the $k$-level Markov condition, once written in terms of conditioning on the future, implies that this ratio is eventually constant in~$n$\,---\, starting with $n=k+1$. Hence, \rgm{s} do generalize stationary $k$-level Markov measures.

Finally, let us discuss Condition~\ref{it:FE} in the context of \rgm{s}. To do so, we will use the fact that the convergence~\eqref{eq:gn-to-g} can also be used to establish the following weak Gibbs condition of Yuri at vanishing topological pressure: there exists an $\Exp{o(n)}$-sequence $(K_n)_{n=1}^\infty$ such that
\[ 
    K_n^{-1} \Exp{\sum_{j=0}^{n-1} \ln g(T^jx) } \leq \P[x_1^n] \leq K_n \Exp{\sum_{j=0}^{n-1} \ln g(T^jx)}
\]
for every $x \in \Omega'$; again, see~\cite[\S{B.3}]{CR23}, but it should be noted that this can be seen as part of the ``g-measure folklore''~\cite{Wa05,OST05,BFV19}.
We are now ready to provide a necessary and sufficient condition on the subshift~$\Omega'$ for~\ref{it:FE} to hold for all \rgm{s} on~$\Omega'$. One special case will be that \rgm{s} on topologically mixing subshifts of finite type with more than one letter satisfy \ref{it:ID} and \ref{it:FE}, allowing for an application of our main result.

\begin{lemma}
    Suppose that $\P$ is a \rgm{} on~$\Omega'$. Then, $\P$ satisfies~\ref{it:FE} if and only if there exists $r$ with the following property: for every $y \in \Omega'$, there exists $t\leq r$ such that $T^ty$ has more than one preimage in~$\Omega'$.
\end{lemma}

\begin{proof}
    Suppose that there exists~$r$ as above.
    Then, for every $y \in \Omega'$, there exists $t \leq r$ such that 
    \[
        g(T^{t-1} y) = 1 - \sum_{\substack{z \in \Omega'\setminus\{T^{t-1} y\} \\ Tz = T^{t}y}} g(z) \leq 1 - \delta,
    \]
    where $\delta := \min g$. This number is positive by continuity and compactness. Therefore,
    \begin{align*}
        \ln g(y) + \ln g(Ty) + \dotsb + \ln g(T^{t-1}y)
        \leq \ln(1-\delta)
    \end{align*}
    and $\ln g(T^{t'}y) \leq \ln(1-\delta)$ for any $t'\geq t$.
    But then, the weak Gibbs property yields 
    \begin{align*}
        \P[y_1^n] 
        &\leq K_n \Exp{\sum_{i=0}^{\lfloor \frac nr\rfloor-1} \sum_{t=0}^{
        r-1} \ln g(T^{ir+t}y)} \\
        &\leq \exp\left(\ln K_n +  \left\lfloor \frac nr \right\rfloor \ln(1-\delta)\right),
    \end{align*}
    with $\ln K_n = o(n)$. We conclude that Condition~\ref{it:FE} holds.
    Suppose now that no such~$r$ exists. Then, for every $n\in\nn$, there exists $y\in\Omega'$ such that $T^ty$ has only one preimage in $\Omega'$ for all $t\leq n$. By the condition~\eqref{eq:g1}, this means that 
    \[
        \ln g(y)+\ln g(Ty)+\dots+\ln g(T^{ n-1} y)=0,
    \]
    which, together with the lower bound in the weak Gibbs property, implies
    \[
        \P[y_1^n]\geq K_n^{-1}=\Exp{-o(n)}.
    \]
    Since the right-hand side is eventually greater than $\Exp{\gamma_+n}$ for any $\gamma_+<0$, \ref{it:FE} fails as well. 
\end{proof}

\subsection{Statistical mechanics}

Let~$\overline{\Omega}'$ be a topologically transitive, two-sided subshift of finite type, and let $\Omega'$ be its one-sided counterpart. Consider a family $(\Phi_X)_{X \Subset \zz}$ of interactions with 
\begin{itemize}
    \item the continuity property $\Phi_X \in C(\overline{\Omega}')$ for all $X \Subset \zz$, with $\Phi_X$ depending on the symbols with indices in the finite subset~$X$ only,
    \item the translation-invariance property $\Phi_{X+1} = \Phi_X \circ \shift$ for all $X \Subset \zz$,
    \item the absolute summability property $\sum_{\substack{X \Subset \zz \\ X \ni 1}} \sup_{x\in\overline{\Omega}'} |\Phi_X(x)|< \infty.$
\end{itemize}
Such interactions are considered e.g.\ in~\cite[\S\S{1.2,\,3.1}]{Rue} and are colloquially said to be in ``the small space''.
It is well known that any equilibrium measure~$\P$ (in the sense of the variational principle) for the energy-per-site potential
\[ 
    \phi := 
    \sum_{\substack{X \Subset \zz \\ \min X = 1}} {\Phi_X}
\]
coming from such a family of interactions
is a translation-invariant Gibbs state in the sense of the Dobrushin--Lanford--Ruelle equations; see e.g.~\cite[\S\S{3.2,\,4.2}]{Rue}.\footnote{With a slight abuse of notation, we are using $\P$ for both the equilibrium measure on~$\overline{\Omega}' \subseteq \cA^\zz$ and its natural restriction to $\Omega' \subseteq \cA^\nn$. Note that, by construction, the potential~$\phi$ only depends on symbols from~$\Omega'$.} Because we are working with a sufficiently regular subshift~$\overline{\Omega}'$, the Dobrushin--Lanford--Ruelle equations and absolute summability can be used to show that~$\P$ satisfies~\ref{it:ID} by adapting the argument of~\cite[\S{9}]{LPS95} for the case $\overline{\Omega}' = \cA^\zz$. Again, the subshift is sufficiently regular for~\ref{it:ID} to yield~\ref{it:KB}; see~\cite[\S\S{3.1,\,B.2}]{CR23}.

We now turn to Condition~\ref{it:FE}, assuming a certain familiarity with the thermodynamic formalism, physical equivalence and the Griffiths--Ruelle theorem on the reader's part; see e.g.~\cite[\S{4}]{Rue}.

\begin{lemma}
    Suppose that $\overline{\Omega}'$, $\Phi$, and $\P$ are as above. If $\Phi$ is not physically equivalent to~$0$ in the sense of Ruelle, then~$\P$ satisfies~\ref{it:FE}.
\end{lemma}

\begin{proof}[Proof sketch.]
    Because we can always add or subtract a constant from each~$\Phi_{\{i\}}$, there is no loss of generality in assuming that $\phi$ has topological pressure $P_{\textnormal{top}}(\phi)=0$. Then, by the weak Gibbs property established e.g. in~\cite[\S{2}]{PS20}, we have
    \[
        \lim_{n\to\infty} \frac 1n \ln \sum_{a\in\supp \P_n} \P[a]^{1-\alpha} = P_{\textnormal{top}}(\phi-\alpha\phi).
    \]
    If $\phi$ is not equivalent to~$0$ in the sense of Ruelle, then the Griffiths--Ruelle theorem guarantees that $\alpha \mapsto P_{\textnormal{top}}(\phi-\alpha\phi)$ is strictly convex; see e.g.~\cite[\S{4.6}]{Rue}. But since this function is easily shown to be nondecreasing, and since it vanishes at~$\alpha = 0$, this implies that $P_{\textnormal{top}}(\phi-\alpha\phi)<0$ for all~$\alpha < 0$. Assuming for the sake of contradiction that~\ref{it:FE} fails, one easily derives a contradiction.
\end{proof}

\begin{remark}
    The converse of this implication does not hold. For example, the uniform measure (measure of maximal entropy) on the full shift for an alphabet with at least two letters arises from vanishing interactions but satisfies \ref{it:FE}.
\end{remark}

Every irreducible, stationary Markov measure
with stochastic matrix~$[P_{a,b}]_{a,b\in\cA}$ can be obtained in this way by considering the following nearest-neighbour interactions on its support:
\[
    \Phi_{\{i,i+1\}}(x) = \ln P_{x_i,x_{i+1}}
\]
for $i\in\nn$ and $\Phi_X(x) = 0$ for $X$ not of the form $\{i,i+1\}$. To see this, one can check by direct computation that, on its support, the Markov measure satisfies the Bowen--Gibbs condition for the corresponding~$\phi$. For $k$-level Markov measures, consider instead
\[
    \Phi_{\{i,\dotsc,i+k-1,i+k\}}(x) = \ln \frac{\P[x_i\dotsc x_{i+k-1}x_{i+k}]}{\P[x_i\dotsc x_{i+k-1}]}.
\]
In this sense, equilibrium measures for potentials arising from interactions that are absolutely summable do generalize stationary $k$-level Markov measures; we refer the reader to~\cite{CHMMP14,BGMMT21} for recent thorough discussions of variants and converses to this observation. This generalization is far reaching as the theory of entropy, large deviations and phase transition is much richer in the small space of interactions than in the space of finite-range interactions.

In a similar vein, equilibrium measures (in the sense of the variational principle on~$\Omega'$) for abstract potentials~$\phi$ in the Bowen class also satisfy~\ref{it:ID}, thanks to the Bowen--Gibbs property; see~\cite[\S{4}]{Wa01}. We refer the reader to~\cite[\S{1}]{Wa01} for a definition of the Bowen class, which can be traced back to~\cite{Bo74}. This class includes potentials with summable variations, and thus H\"older-continuous potentials, and thus potentials naturally associated to stationary $k$-level Markov measures. A more complete discussion from the point of view of decoupling\,---\,including relaxation of the conditions on~$\Omega'$\,---\,can be found in~\cite[\S{2.3}]{CR23}.

\subsection{Hidden-Markov measures}\label{ssec:hmm}

While the above generalizations beyond Markovianity are often studied in the literature on mathematical physics and abstract dynamical systems, they might not be the most natural from an information-theoretic point of view; hidden-Markov models would most likely come to mind first for many practitioners.
We recall that, among several equivalent representations, a stationary hidden-Markov measure~$\P$ can be characterized by a tuple $(\pi, P, R)$ where $(\pi, P)$ characterizes in the usual way a stationary Markov process on a set~$\mathcal{S}$, called the \emph{hidden alphabet}, and~$R$ is a $(\#\mathcal{S})$-by-$(\#\cA)$ matrix whose rows each sum to~1:
\begin{equation*}
    \P[a_1^n] = \sum_{s_1^n \in \mathcal{S}^n} \pi_{s_1} R_{s_1, a_1} P_{s_1, s_2} R_{s_2, a_2}\cdots P_{s_{n-1}, s_{n}} R_{s_n, a_n}
\end{equation*}
for $n \in \nn$ and $a_1^n \in \cA^n$.
We restrict our attention to the case where $\mathcal{S}$ is a finite set and $P$ is irreducible.
We view the entry $R_{s,a}$ as the probability of observing $a \in \cA$ at a given time step given the \emph{hidden state} $s \in \mathcal{S}$ at that same time step\,---\,the dynamics of the latter governed by the hidden-Markov chain~$(\pi,P)$. 
There exist only very singular examples of such measures for which \ref{it:FE} fails. As exhibited by our next lemma, this can only happen if the process is eventually almost-surely deterministic.

\begin{lemma}
    Let $\P$ be as above. Then, $\P$ satisfies~\ref{it:FE} if and only if, for each $s \in \mathcal{S}$, there exists $L$ such that 
    \[
        \#\{a \in \cA^L : \P[a | s_1 = s] > 0 \} > 1.
    \]
\end{lemma}

\begin{proof}
    Suppose that for each $s\in \mathcal S$ there exists $L$ as above. By inspection of the canonical form of $P$ provided by the Perron--Frobenius theorem, one deduces that there exists a finite set $\Sigma'$ of possible row vectors~$\sigma$ that can arise as limit points for sequences of the form  $([P^{m}]_{i,\cdot\,})_{m=1}^\infty$. Let $\Sigma := \Sigma' \cup \{\pi\}$ with $\pi$ the unique invariant probability row vector for~$P$. By stochasticity, each $\sigma \in \Sigma$ has nonnegative entries that sum to 1.
    In this context, {by assumption}, there exists $L \in \nn$ such that 
    \begin{align*}
        \delta := \max_{a \in \supp \P_L}\max_{\sigma \in\Sigma} \sum_{s \in \mathcal{S}^L} \sigma_{s_1} P_{s_1,s_2} \cdots P_{s_{L-1},s_{L}} R_{s_1,a_1} \cdots R_{s_{L},a_{L}} 
        &= \max_{a \in \supp \P_L} \max_{\sigma \in \Sigma} \sum_{s \in \mathcal{S}} \sigma_{s} \P[a|s_1 = s]
    \end{align*}
    is strictly less than~$1$.
    Given $\epsilon > 0$, by inspection of the same canonical form, there exists $m \in \nn$ with the following property: for all~$i$, there is $\sigma \in \Sigma$ such that 
    \[
        [P^m]_{i,\,\cdot} <  \sigma + \epsilon
    \]
    Then, taking $a \in \supp \P$ and $n \geq L$, 
    \begin{align*}
        \P[a_1^n] 
            &\leq \P[a_1^{L + q(m+L)}]
    \end{align*}
    for $q := \max \{k \in \nn_0 : n \geq L + k(m+L) \}$.
    We introduce the shorthands $\mathcal{R}_0(s) = R_{s_1, a_1} \cdots R_{s_L, a_L}$,
    \[
        \mathcal{R}_k(s) = R_{s_{(k-1)(m+L)+L+1}, a_{(k-1)(m+L)+L+1}} \cdots R_{s_{k(m+L)+L}, a_{k(m+L)+L}}
    \]
    and
    \[
        \mathcal{R}'_k(s) = R_{s_{k(m+L) +1}, a_{k(m+L) +1}} \cdots R_{s_{k(m+L)+L}, a_{k(m+L)+L}}
    \]
    when $1 \leq k \leq q$. We also identify $s^1_0 \equiv s_L$, $s^2_0 \equiv s^1_{m+L}$, $s^3_0 \equiv s^2_{m+L}$, and so on, and so forth. One then obtains:
    \begin{align*}
        &\P[a_1^{L + q(m+L)}]\\
            &\quad = 
            \sum_{\substack{s_1, \dotsc, s_L \\ s_1^{k}, \dotsc, s_{m+L}^k \\ \text{for } 1 \leq k \leq q}}
                \pi_{s_1} P_{s_1, s_2} \cdots P_{s_{L-1}, s_{L}} \mathcal{R}_0(s) \prod_{k=1}^q P_{s^{k}_{0}, s^k_1}P_{s^{k}_{1}, s^{k}_2} \dotsb P_{s^{k}_{m+L-1}, s^{k}_{m+L}} \mathcal{R}_k(s) \\
            &\quad \leq 
            \sum_{\substack{s_1, \dotsc, s_L \\ s_{m}^{k}, \dotsc, s_{m+L}^k \\ \text{for } 1 \leq k \leq q}}
                \pi_{s_1} P_{s_1, s_2} \cdots P_{s_{L-1}, s_{L}} \mathcal{R}_0(s)
                \prod_{k=1}^q (\sigma_{s_m^k}^{(s_0^k)} + \epsilon)P_{s^{k}_{m}, s^{k}_{m+1}} \dotsb P_{s^{k}_{m+L-1}, s^{k}_{m+L}} \mathcal{R}_k'(s)
    \end{align*}
    for some appropriate choices of $\sigma^{(s^{k}_0)} \in \Sigma$ that depend on~$m$ and the index $s^{k}_0$ only. Therefore, 
    \begin{align*}
        \P[a_1^{L + q(m+L)}]
            &\leq \delta \cdot (\delta + \epsilon (\#\mathcal{S}))^{q}.
    \end{align*}
    By taking $\epsilon>0$ such that $\delta + \epsilon (\#\mathcal{S}) < 1$ and noting that $q$ scales linearly with $n$, \ref{it:FE} holds.

    To see the converse implication, suppose that there exists $t\in\mathcal S$ such that there is no $L$ as above. Then, there exists $a \in \Omega$ such that $\P[a_1^n|s_1 = t] = 1$ for all $n \in \nn$. Since $$\P[a_1^n] \geq \P[a_1^n | s_1 = t] \cdot \pi_t = \pi_t$$ for all $n \in \nn$ and $\Exp{\gamma_{+}n}$ is eventually smaller than $\pi_t > 0$ for all $\gamma_{+} < 0$, \ref{it:FE} fails.
\end{proof}

One can show that every stationary hidden-Markov measure satisfies the upper bound in \ref{it:ID}. But in general,\,---\,even if $P$ is irreducible\,---\,only a weaker form of the lower bound, known as \emph{selective lower decoupling}, holds; see~\cite[\S{2}]{BCJPPExamples} and~\cite[\S{2}]{CJPS19}. The fact that selective lower decoupling implies \ref{it:KB} but does not imply the condition called~{\ref{it:Ad}} in Section~\ref{ssec:comments} seems to pose a genuine obstacle. Determining whether the ZM estimation remains generally valid in the class of irreducible, hidden-Markov measures remains\,---\,to our knowledge\,---\,an important open problem.

In the further specialized case where the elements of~$R$ are all in~$\{0,1\}$\,---\,this is sometimes called the \emph{function-Markov} or \emph{lumped-Markov} case\,---, some conditions for the $g$-measure property (and thus \ref{it:ID}) are discussed in~\cite{CU03,Yo10,Ve11}. However, it is not difficult to find examples for which none of these known sufficient conditions hold\begin{figure}[b]
    \centering
    \includegraphics{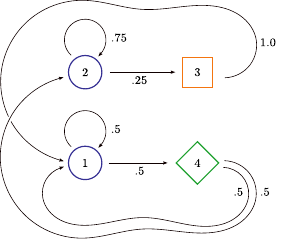}
    \caption{An example that does not satisfy {\ref{it:Ad}}: a direct computation shows that $[{\color{newdarkblue}\circ}{\color{newdarkblue}\circ} \dotsc {\color{newdarkblue}\circ}{\color{newgreen}\diamond}]$ is too unlikely compared to $[{\color{newdarkblue}\circ}{\color{newdarkblue}\circ} \dotsc {\color{newdarkblue}\circ}]$.}
    \label{fig:hmm-example}
    \end{figure}.
\begin{example}
\newcommand{\ssq}{\mathbin{\rotatebox[origin=c]{-45}{$\diamond$}}}
    The stationary measure on~$\{{\color{newdarkblue}\circ}, {\color{newgreen}\diamond},{\color{orange}\ssq}\}^\nn$ built from the four-hidden-state chain depicted in Figure~\ref{fig:hmm-example} satisfies the upper bound in~\ref{it:ID}, as well as \ref{it:FE} and \ref{it:SE}, but not~\ref{it:Ad}\,---\,and hence not~\ref{it:ID}.
\end{example}

\paragraph*{Acknowledgments.} The authors would like to thank G.\ Cristadoro, N.\ Cuneo and V.\ Jak\v{s}i\'{c} for stimulating discussions on the topic of this note.  The research of NB and RR was partially funded by the \emph{Fonds de recherche du Qu\'ebec\,---\,Nature et technologies} (FRQNT) and by the Natural Sciences and Engineering Research Council of Canada (NSERC). The research of RG was partially funded by the {Rubin Gruber Science Undergraduate Research Award} and {Axel W Hundemer}. The research of GP was done under the auspices of the \emph{Gruppo Nazionale di Fisica Matematica} (GNFM) section of the \emph{Istituto Nazionale di Alta Matematica} (INdAM). Part of this work was done during a stay of the four authors in Neuville-sur-Oise, funded by CY Initiative (grant \emph{Investissements d'avenir} ANR-16-IDEX-0008).

\newcommand{\etalchar}[1]{$^{#1}$}

\end{document}